\documentclass[11pt]{llncs}

\usepackage[noend]{algorithmic}
\usepackage[T1]{fontenc}

\usepackage{graphicx}

\newcommand{\comment}[1]{}

\begin{document}
\title{Detecting Superbubbles in Assembly Graphs}
\author{
  Taku Onodera\inst{1}
\and 
  Kunihiko Sadakane\inst{2}
\and
  Tetsuo Shibuya\inst{1}
}
\authorrunning{T. Onodera et al.}

\institute{
  Human Genome Center, Institute of Medical Science, University of Tokyo
4-6-1 Shirokanedai, Minato-ku, Tokyo 108-8639, Japan.
  \email{\{tk-ono,tshibuya\}@hgc.jp}
\and
  National Institute of Informatics, 2-1-2 Hitotsubashi, Chiyoda-ku, \\
  Tokyo 101-8430, Japan.
  \email{sada@nii.ac.jp}
}

\maketitle

\begin{abstract}
We introduce a new concept of a subgraph class called a superbubble for 
analyzing assembly graphs, 
and propose an efficient algorithm for detecting it.
Most assembly algorithms utilize assembly graphs like 
the de Bruijn graph or the overlap graph constructed from reads. 
From these graphs, 
many assembly algorithms first detect simple local graph structures (motifs), 
such as tips and bubbles, mainly to find sequencing errors. 
These motifs are easy to detect, 
but they are sometimes too simple to deal with more complex errors. 
The superbubble is an extension of the bubble, which is also 
important for analyzing assembly graphs. 
Though superbubbles are much more complex than ordinary bubbles, 
we show that they can be efficiently enumerated. 
We propose an average-case linear time algorithm 
({\it i.e.}, $O(n+m)$ for a graph with $n$ vertices and $m$ edges) 
for graphs with a reasonable model, 
though the worst-case time complexity of our algorithm 
is quadratic ({\it i.e.}, $O(n(n+m))$). 
Moreover, the algorithm is practically very fast: 
Our experiments show that 
our algorithm runs in reasonable time with a single CPU core 
even against a very large graph of a whole human genome. 
\end{abstract}


\section{Introduction}
\label{sec:introduction}
The sequencing technologies have evolved dramatically in the past 25 
years, and nowadays many 
next-generation sequencers (NGSs) can sequence a human genome-size genome 
in only a few hours with very small costs. 
But still there is no sequencing technology that can sequence the 
entire genome at a time without breaking the genome into 
millions or billions of short reads. Thus assembling these reads into 
a whole genome has been one of the most important computational 
problems in molecular biology, and quite a few algorithms have been 
proposed for the problem~\cite{KasMor06,MilKor10,Pop09} 
despite the computational difficulty of the problem~\cite{Mye95}. 

Most assembly algorithms construct some graph 
in their first stage. They are 
categorized into two types depending on the types of the graph. 
Many old-time assemblers utilize a graph called the {\it overlap graph}, 
in which a vertex corresponds to a read and an edge corresponds to 
a pair of reads that have an enough-length overlap~\cite{BatJaf02,HuaYan05,MyeSut00}. 
More recent algorithms often utilize a graph called 
the {\it de Bruijn graph}, in which an edge corresponds to 
a $k$-mer that exists in reads and a vertex corresponds to 
the shared $(k-1)$-mer between the adjacent 
$k$-mers~\cite{JacReg10,LiZhu10,MacPrz09,PevTan01,SahShi12,SimWon09,ZerBir08}. 
The de Bruijn graph is said to be more suitable for NGS short reads of 
large depth. 

The next step of most sequencing algorithms after constructed the graph 
is to simplify the obtained graph by decomposing a maximal unbranched 
sequence of edges (which is called a {\it unipath}) 
into one single edge~\cite{JacReg10,MacPrz09,SahShi12} (Fig.~\ref{fig:unipath}). 
The obtained graph is called a {\it unipath 
graph}. 
After obtained the unipath graph, 
many sequencing algorithms next
detect simple typical motif structures caused by errors to detect errors: 
The most common motifs are tips, bubbles, and 
cross links~\cite{JacReg10,LiZhu10,SahShi12,ZerBir08} (Fig.~\ref{fig:motifs}). 

\begin{figure}[bt]
\begin{center}
  \includegraphics[scale=0.6]{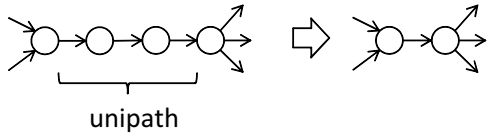}
\caption{Construction of a unipath graph.}
\label{fig:unipath}
\end{center}
\end{figure}

\begin{figure}[bt]
\begin{center}
\begin{minipage}{0.275\textwidth}
\begin{center}
  \includegraphics[scale=0.5]{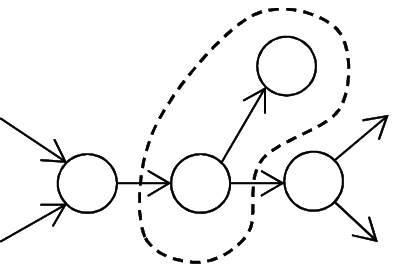}
\end{center}
\end{minipage}
\begin{minipage}{0.275\textwidth}
\begin{center}
  \includegraphics[scale=0.5]{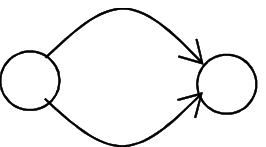}
\end{center}
\end{minipage}
\begin{minipage}{0.275\textwidth}
\begin{center}
  \includegraphics[scale=0.5]{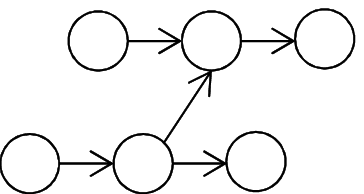}
\end{center}
\end{minipage}
\begin{minipage}{0.275\textwidth}
\begin{center}
\ \\ 
(1) A tip
\end{center}
\end{minipage}
\begin{minipage}{0.275\textwidth}
\begin{center}
\ \\ 
(2) A bubble
\end{center}
\end{minipage}
\begin{minipage}{0.275\textwidth}
\begin{center}
\ \\ 
(3) A cross link
\end{center}
\end{minipage}\vspace{-1ex}
\caption{Assembly graph simple motifs.}
\label{fig:motifs}
\end{center}
\end{figure}

\setcounter{footnote}{0}
A tip (Fig.~\ref{fig:motifs} (1)) 
is a low-frequency edge whose end (or start) vertex has no outgoing 
(resp. incoming) edges, 
which goes out from (resp. comes into) a high-frequency 
vertex\footnote{We say 'low/high'-frequency vertices/edges for vertices/edges 
that correspond to few/many reads.}.  
This motif often appears in case there are some error(s) around the end of 
a read. A bubble (Fig.~\ref{fig:motifs} (2)) consists of 
multiple edges (with the same direction) between a pair of vertices, 
which is often caused by error(s) somewhere in the middle of a read. 
A cross link (Fig.~\ref{fig:motifs} (3)) 
is a low-frequency edge that lies between high-frequency vertices. 
This appears when a substring of a read accidentally becomes (by error) 
the same substring that appears in a different region. 
All of these motifs are easy to find (obviously in linear time) 
due to their simplicity. 

\begin{figure}[bt]
\begin{center}
  \includegraphics[scale=0.45]{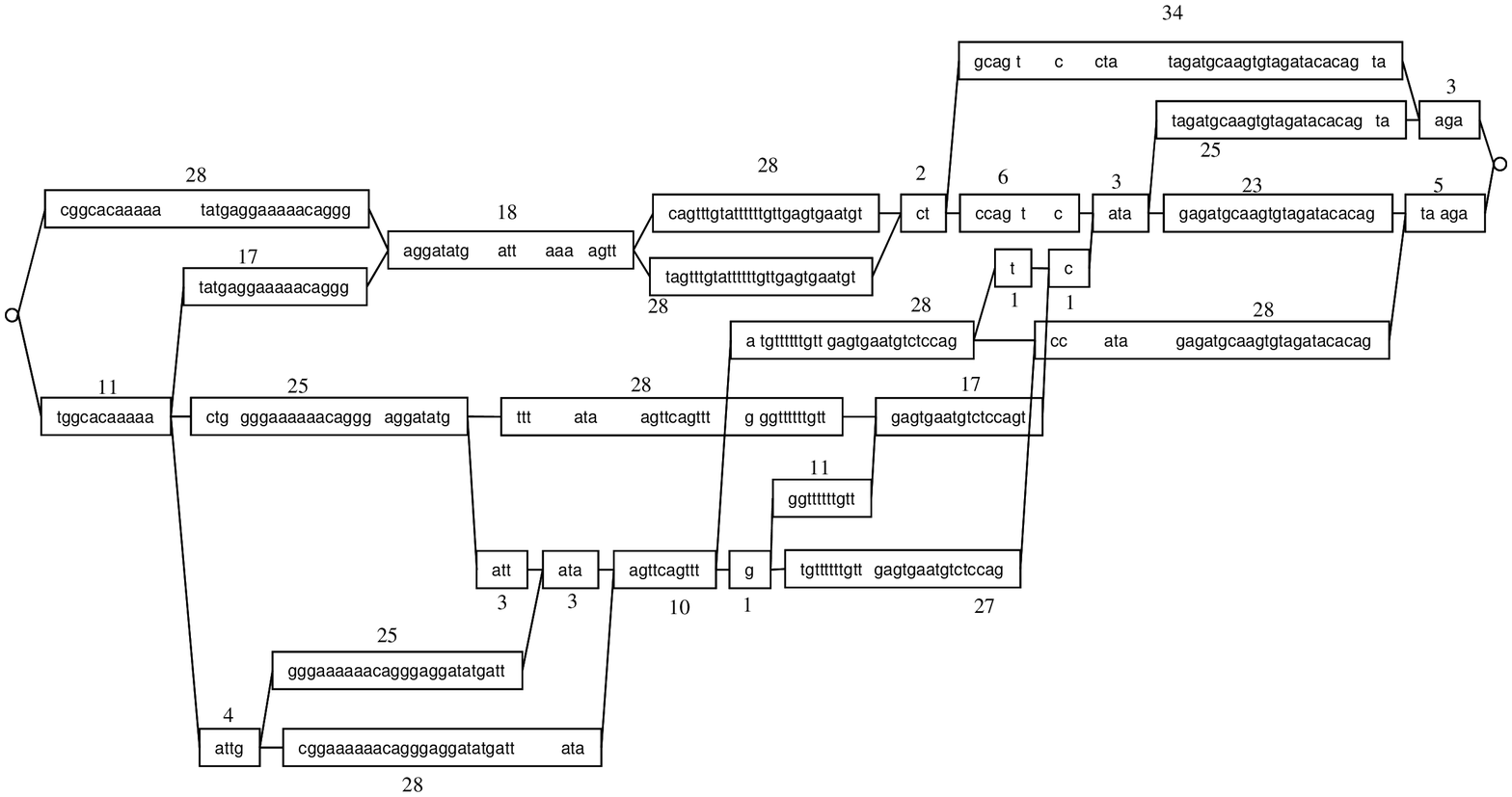}
\caption{A superbubble: A very complicated structure caused by errors or repeats. All the edges are labeled with sequences (vertices are not shown). 
The gaps in the labels are inserted manually in the figure 
to show alignment between edge 
labels that start at different offsets from the entrance of the superbubble.}
\label{fig:superbubble}
\end{center}
\end{figure}

But we should consider much more complex structures 
if input reads are erroneous (as in the case of the third generation 
sequencers), 
have many repeats (as in many large-scale genomes/meta-genomes), or have 
many mutations (as in cancer genomes). 
Fig.~\ref{fig:superbubble}
shows an example of a subgraph of a unipath graph obtained from actual 
whole human genome reads (the same set of reads 
used in the experiments in section~\ref{sec:experiments}). 
In this subgraph, paths from the leftmost vertex branch 
to many paths but they converge into the rightmost single vertex in the end, 
and there are no cycles in this subgraph, 
{\it i.e.}, the subgraph forms a directed acyclic graph (DAG). 
The vertices between the leftmost vertex and the rightmost vertex has 
no outgoing/incoming edges to/from external vertices ({\it i.e.}, 
vertices not in this subgraph). 
An important point is that all the paths have similar labels with 
similar lengths.\footnote{The experiments in section \ref{sec:experiments} will show that the path label
lengths of a superbubble are only at most $5$\% different in more than $85$\% of the detected superbubbles.} 
We call this kind of a subgraph a {\it superbubble}, as it can be 
considered as an extension of an ordinary simple bubble 
(more detailed definition of superbubbles will be 
given in section~\ref{sec:preliminaries}). 
Superbubbles are complicated, but it is apparent that 
many of them are formed as a result of errors, inexact repeats, 
diploid/polyploid genomes, or frequent mutations. 
Thus detection of superbubbles should be very important, and 
it should be useful if we can detect them efficiently. 
For example, further 
time-consuming complicated algorithms (e.g., optimal alignment, 
paired-end read analyses, etc) 
are applicable against the superbubbles, 
even if they are too complicated to use against the entire graph. 

In the followings, we will give detailed 
definition of the superbubbles in section~\ref{sec:preliminaries}, 
and show an efficient algorithm for
finding superbubbles in section~\ref{sec:algorithm}. 
We will show that the algorithm 
runs in average-case linear time against graphs with a reasonable 
model, though the worst-case time complexity is quadratic. 
In section~\ref{sec:experiments}, we will show that the superbubbles can be 
efficiently enumerated in reasonable time with a small machine, 
through large-scale experiments against reads from a whole human genome.

\section{Preliminaries}
\label{sec:preliminaries}

\subsection{Superbubble}
\label{sec:superbubble}

Here, we formally define superbubbles and show some properties of them which are necessary in the rest of the paper.

\begin{definition}\label{def:superbubble}
Let $G=(V,E)$ be a directed graph.
If an ordered pair of distinct vertices $(s,t)$ satisfies the following:
\begin{description}
\item[reachability]\label{cond:reachability} $t$ is reachable from $s$;
\item[matching]\label{cond:matching} the set of vertices reachable from $s$ without passing\footnote{Passing through a vertex means that visiting and then leaving it, not just visiting or leaving alone.} through $t$ is equal to the set of vertices from which $t$ is reachable without passing through $s$;
\item[acyclicity]\label{cond:acyclicity} the subgraph induced by $U$ is acyclic where $U$ is the set of vertices in the above condition;
\item[minimality]\label{cond:maximality} no vertex in $U$ other than $t$ forms a pair with $s$ that satisfies the conditions above,
\end{description}
then we say that the subgraph in the description of the acyclicity condition is a \textbf{superbubble} and $s$, $t$ and $U\setminus\{s,t\}$ are this superbubble's \textbf{entrance}, \textbf{exit} and \textbf{interior} respectively.
For any pair of vertices $(s,t)$ that satisfies the above conditions, we denote the superbubble as $\langle s,t\rangle$.
\end{definition}

To take full advantage of the notation $\langle s,t\rangle$, we first need to confirm that if $(s_1,t_1)\neq(s_2,t_2)$ then $\langle s_1,t_1\rangle\neq\langle s_2,t_2\rangle$.
The following remark ensures it.
\begin{remark}
There is a one-to-one correspondence between the vertex pairs satisfying the conditions in Definition \ref{def:superbubble} and superbubbles.
\end{remark}
\begin{proof}
Because of the acyclicity condition, the vertices of a superbubble can be topologically sorted, {\it i.e.}, they can be ordered in such a way that if $v$ is reachable from $u$ then $u<v$. Due to the matching condition, $s$ (resp. $t$) is the minimum (resp. maximum) ordered vertex.
\end{proof}

Now we observe a proposition which clarifies the situation and motivates linear time enumeration of superbubbles.

\begin{proposition}
Any vertex can be the entrance (resp. exit) of at most one superbubble.
\end{proposition}

Note that this proposition does not exclude the possibility that a vertex is the entrance of a superbubble and the exit of another superbubble.

\begin{proof}
We prove the proposition by \textit{reductio ad absurdum}.
Suppose $\langle s,t_1\rangle$ and $\langle s,t_2\rangle$ are distinct superbubbles.
If $t_2$ is a vertex in $\langle s,t_1\rangle$, then $t_2$ is in the interior of $\langle s,t_1\rangle$ but this contradicts to the minimality condition for $\langle s,t_1\rangle$.
Similarly, $t_1$ being a vertex in $\langle s,t_2\rangle$ also results in a contradiction.

Suppose, on the other hand, that $t_2$ is not a vertex in $\langle s,t_1\rangle$. There is a path from $s$ to $t_2$.
By removing cycles from $t_2$ to $t_2$ if necessary, this path can be taken in such a way that $t_2$ appears only at the last step and at this time, all vertices in the path are in $\langle s,t_2\rangle$.
On the other hand, the vertex just before the first vertex on the path that is not in $\langle s,t_1\rangle$ is $t_1$.
In particular this means that $t_1$ is in $\langle s,t_2\rangle$ but this leads to contradiction by the first half of the argument.
\end{proof}

\begin{corollary}
There are $O(n)$ superbubbles in a graph with $n$ vertices.
\end{corollary}

Before closing this subsection, let us point out, without proof, yet another property of superbubbles that is not directly necessary for this work but worth mentioning to grasp the picture.

\begin{claim}
If two distinct superbubbles share a vertex, either one's exit is the other's entrance or one is included in the other's interior.
\end{claim}

\subsection{Construction of a Unipath Graph}
\label{sec:graph}
Given a set ${\cal R}$ of reads, we first construct the de Bruijn graph~\cite{PevTan01}.
Let $T = T[1,m]$ be a read of length $m$ in ${\cal R}$.
The $k$-mers of $T$ are length-$k$ substrings of $T$, that is,
$T[i,i+k-1]$ for $i = 1,2,\ldots,m-k+1$.
Let $K$ denote the multiset of $k$-mers of all reads in ${\cal R}$,
and $K_d$ denote the set of (distinct) $k$-mers that appear at least $d$ times in $K$.
A $k$-mer in $K_d$ is called a \emph{solid} $k$-mer.

The de Bruijn graph $G=(V,E)$ of ${\cal R}$ is defined as follows.
The vertex set $V$ is the set of $(k-1)$-mers defined as
$V = \{ T[1,k-1] \mid T[1,k] \in K_d \} \cup \{ T[2,k] \mid T[1,k] \in K_d \}$.
The edge set $E$ is defined as $\{ (u,v) \mid  \exists T[1,k] \in K_d, u=T[1,k-1], v=T[2,k]  \}$.
The edge label of $(u,v)$ is $T[k]$ if $u=T[1,k-1], v=T[2,k]$.
Typical values of $k$ and $d$ are $k=28$, $d=3$.

\def\outgoing{\textit{outgoing}}
\def\incoming{\textit{incoming}}
\def\outdeg{\textit{outdeg}}
\def\indeg{\textit{indeg}}
\def\last{\textit{last}}
\def\Node{\textit{Node}}
\def\fwd{\textit{fwd}}
\def\bwd{\textit{bwd}}

\begin{figure}[bt]
\begin{center}
  \includegraphics[scale=0.29]{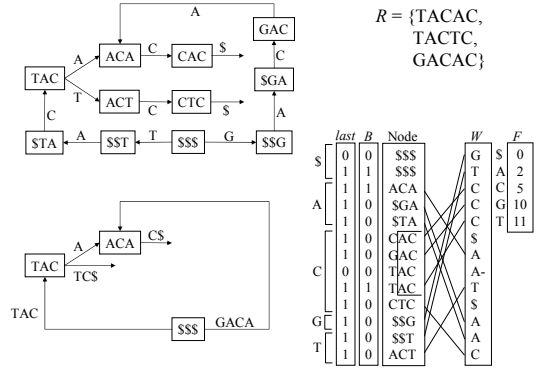}
\caption{Top right: The input set ${\cal R}$, top left: The de Bruijn graph of ${\cal R}$ with $k=3, d=1$,
bottom left: the unipath graph, bottom right: the succinct de Bruijn graph and the unipath graph. 
Non-branching nodes are removed. 
We store only {\last}, $B$, $W$ and $F$.}
\label{fig:debruijn}
\end{center}
\end{figure}

We use the succinct de Bruijn graph~\cite{BowOno12}, which is a compressed representation
of the de Bruijn graph of ${\cal R}$.  For a set of $m$ solid $k$-mers,
the succinct de Bruijn graph uses $4m +o(m)$ bits to encode the graph,
and supports the following operations.
\begin{itemize}
\item ${\outdeg}(v)$/${\indeg}(v)$ returns the number of outgoing/incoming edges from/to vertex $v$ in $O(1)$ time, respectively.
\item ${\outgoing}(v, c)$ returns the vertex $w$ pointed to by the outgoing edge of vertex $v$
with edge label $c$ in $O(1)$ time.  If no such vertex exists, it returns $-1$.
\item ${\incoming}(v, c)$ returns the vertex $w = T[1,k-1]$ such that 
there is an edge from $w$ and $v$ and $T[1] = c$ in $O(k)$ time.
If no such vertex exists, it returns $-1$.
\end{itemize}

From a de Bruijn graph $G=(V,E)$, we construct a unipath graph $G'=(V',E')$ as follows.
The vertex set $V'$ is a subset of $V$ such that any vertex in $V'$ has more than
one outgoing edges or more than one incoming edges.
The edge set $E'$ is the multiset of all pairs $(u,v)$ such that $u, v \in V'$ and
there is a path $u,x_1,x_2,\ldots,x_\ell,v$ in $G$ and outdegrees and indegrees
of $x_1,x_2,\ldots,x_\ell$ are all one.  The edge label of $(u,v)$ is the concatenation
of edge labels of $(u,x_1), (x_1,x_2),\ldots,(x_{\ell-1},x_\ell),(x_\ell,v)$ in $G$.
The length of the edge label is $\ell+1$.

In addition to the data structure of the succinct de Bruijn graph,
we use a bit vector $B[1,m]$ where $m = | E |$ is the number of edges in $G$
to represent the unipath graph $G'$.
We set $B[v] = 1$ if and only if the vertex $v$ of $G$ is also a vertex of $G'$.
The outdegree and the indegree of $v$ in $G'$ is equal to those of $v$ in $G$.
To find the vertex ${\outgoing}(v, c)$ in $G'$, we first compute $w = {\outgoing}(v, c)$ in $G$.
Then we repeatedly traverse the unique outgoing edge of $w$ until $B[w] = 1$.
The resulting vertex is the answer.
The unipath graph is constructed in linear time from the succinct de Bruijn graph
because each of the {\outdeg}, {\indeg}, and {\outgoing} operations takes constant time.
Figure~\ref{fig:debruijn} shows an example.

\section{Algorithm}
\label{sec:algorithm}
Here, we explain how to enumerate all superbubbles in a given graph.
As we have seen in subsection \ref{sec:superbubble}, each vertex can be the entrance of at most one superbubble.
Therefore, once we have a way to check if a vertex $s$ has another vertex $t$ s.t. $(s,t)$ is the entrance/exit pair, then we can find all superbubbles just by iterating this procedure for all $s\in V$.
Below, we focus our attention on this reduced problem.

\paragraph{Description} 
The algorithm is based on the standard topological sorting.
It takes a directed graph $G=(V,E)$ and $s\in V$ as inputs, and returns $t\in V$ s.t. $(s,t)$ is an entrance/exit pair of a superbubble if any.
It proceeds by visiting vertices one by one maintaining the dynamic set $S$ of vertices it can visit the next time.
Initially, $S$ is set to be $\{s\}$.
It also maintains a label for each vertex.
The label \textsf{visited} means that the vertex has already been visited.
The label \textsf{seen} means that the vertex has at least one \textsf{visited} parent
.
At each step, the algorithm picks out an arbitrary vertex $v$ from $S$ labeling it as \textsf{visited} and label each child as \textsf{seen}.
If all the parents of a child are \textsf{visited}, it pushes the child into $S$.
In visiting vertices, the algorithm aborts anytime when it finds a vertex with no child, which means a tip, or a parent of $s$, which means a cycle because any vertex visited is a descendent of $s$.
After visiting a vertex, the algorithm tests if it is going to visit the exit at the next step as follows.
First it checks if $S$ consists of one vertex, say $t$, and no vertex other than $t$ is labelled as \textsf{seen}.
If not, the test is negative.
Otherwise, the algorithm further checks if the edge $(t,s)$ exists or not.
If it does, the algorithm aborts because it just found a path from $s$ to $s$, a cycle.
Otherwise, the algorithm returns $t$.
The algorithm aborts if $S$ runs out.
\begin{figure}[tbh]
\begin{algorithmic}[1]
\REQUIRE directed graph $G=(V,E), s\in V$
\ENSURE returns $t$ s.t. $(s,t)$ is an entrance/exit pair of a superbubble if any
\STATE push $s$ into $S$
\REPEAT
\STATE pick out an arbitrary $v\in S$
\STATE label $v$ as \textsf{visited}
\IF{$v$ does not have a child}
\STATE abort // tip
\ENDIF
\FOR{$u$ in $v$'s children}
\IF{$u=s$}
\STATE abort // cycle including $s$
\ENDIF
\STATE label $u$ as \textsf{seen}
\IF{all of $u$'s parents are \textsf{visited}}
\STATE push $u$ into $S$
\ENDIF
\ENDFOR
\IF{only one vertex $t$ is left in $S$ and no other vertex is \textsf{seen}}
\IF{edge $(t,s)$ does not exist}
\RETURN $t$
\ELSE
\STATE abort // cycle including $s$
\ENDIF
\ENDIF
\UNTIL{$|S|=0$}
\end{algorithmic}
\caption{Pseudocode of an algorithm to find the corresponding exit of an potential entrance}
\label{fig:algorithm}
\end{figure}

\paragraph{Correctness}
A vertex can be pushed into $S$ at most once because it happens when all its parents are \textsf{visited} and once \textsf{visited} a vertex never cease to being so.
Thus, the algorithm can pick out a vertex from $S$ at most $n$ times and in particular it halts.
Below, we prove the correctness of the returned value, which reduces to the followings: a) if the input vertex is the entrance of some superbubble, then the algorithm returns the corresponding exit; b) if the algorithm returns a vertex, it is the exit of a superbubble and the input vertex is the corresponding entrance.

First, we observe an invariant.
Let $V_\mathsf{seen}$ be the set of vertices labelled as \textsf{seen} and $V_\mathsf{visited}$ be the set of vertices labelled as \textsf{visited}.
Let $V_\mathrm{to}$ be the set of vertices that are reachable from $s$ without passing through any element of $V_\mathsf{seen}$ and let $V_\mathrm{from}$ be the set of vertices from which at least one element of $V_\mathsf{visited}\cup S$ can be reachable without passing through $s$.
\begin{lemma}
After the algorithm visits a vertex, i.e., after the line 12 of the pseudocode in Figure~\ref{fig:algorithm} is executed, $V_\mathrm{to} = V_\mathsf{visited}\cup V_\mathsf{seen}$ and $V_\mathrm{from} = V_\mathsf{visited}\cup S$.
In particular, if the algorithm returns $t$, then $(s,t)$ satisfies the matching condition.
\end{lemma}
\begin{proof}
We prove the first half by mathematical induction.
After the first visit, $V_\mathsf{visited}$, $V_\mathsf{seen}$ and $S$ consist of $s$, $s$'s children and $s$'s children with indegree 1 respectively and the lemma holds.
Suppose the lemma holds up to the visit to some vertex.
During the visit to the next vertex, say $v$,
\begin{enumerate}
\item $v$ is removed from $S$ and its label is changed from \textsf{seen} to \textsf{visited};
\item all children of $v$ are labelled as \textsf{seen};
\item the children of $v$ whose parents are all \textsf{visited} are added to $S$.
\end{enumerate}
Consequently, both $V_\mathrm{to}$ and $V_\mathsf{visited}\cup V_\mathsf{seen}$ acquire the vertices reachable from $v$ without passing through any element of $V_\mathsf{seen}$, i.e., the children of $v$.
Therefore, $V_\mathrm{to} = V_\mathsf{visited}\cup V_\mathsf{seen}$ still holds.
On the other hand, $V_\mathsf{visited}\cup S$ acquires the vertices newly added to $S$, i.e., the children of $v$ whose parents are all labelled as \textsf{visited}.
Now these vertices are also in $V_\mathrm{from}$ because $V_\mathrm{from}\supseteq V_\mathsf{visited}\cup S$ by definition.
Furthermore, they are the only vertices $V_\mathrm{from}$ acquires because the parents of them were already in $V_\mathrm{from}$ after the previous visit by the induction hypothesis.
Therefore, $V_\mathrm{from} = V_\mathsf{visited}\cup S$ also stays true.

Next, we prove the last half.
If the algorithm returns $t$, after the last visit, $V_\mathrm{to} = V_\mathrm{from}$ because $S = V_\mathsf{seen}$ due to the first half.
On the other hand, at this time, $V_\mathrm{to}$ consists of the vertices reachable from $s$ without passing through $t$ because $V_\mathsf{seen} = \{t\}$.
Therefore, it suffices to show that $V_\mathrm{from}$ consists of the vertices from which $t$ is reachable without passing through $s$.
This is true because after every visit, from any vertex in $V_\mathsf{visited}$ at least one vertex in $V_\mathsf{seen}$ is reachable without passing through $s$, a fact which can be proven easily by mathematical induction again.
\end{proof}

Next, we prove a).
Let $t$ be the exit corresponding to $s$.
Because of the matching condition of $(s,t)$, the algorithm never aborts due to a tip or running out of $S$ at least up to the point when $t$ is pushed into $S$, no matter if $t$ is pushed into $S$ at all.
Similarly, the algorithm never aborts due to a cycle up to the same point because of the acyclicity condition of $(s,t)$.
On the other hand, if $t$ is indeed pushed into $S$, then $t$ must be the only vertex \textsf{seen} and all other vertices of $\langle s,t\rangle$ must be \textsf{visited} due to the matching condition of $(s,t)$ and the lemma.
Therefore, the only possibilities left are that the algorithm outputs $t$ or some other vertex in $\langle s,t\rangle$.
But the second case never happens because a vertex, say $v$, other than $t$ in $\langle s,t\rangle$ is output, then the pair $(s,v)$ satisfies the reachability, matching (due to the lemma) and acyclicity conditions, which contradicts to the minimality condition of $(s,t)$.

Last, we prove b).
Suppose the algorithm returns a vertex $t$.
Obviously, $t$ is reachable from $s$.
The matching condition holds because of the lemma.
The alleged superbubble does not contain cycles including $s$ because otherwise the algorithm must have aborted.
And it does not contain cycles not including $s$ because otherwise the first vertex visited in the cycle has a parent in the cycle.
This means the parent has been visited earlier, which contradicts the way the child was chosen.
Thus, the acyclicity condition holds.
The minimality condition holds because otherwise, there is a vertex $v$ s.t. $(s,v)$ is an entrance/exit pair and because of a) the algorithm must have returned $v$, instead of $t$.

\paragraph{Analysis}
In the worst case, each execution of the algorithm takes \mbox{$\Theta(n+m)$-time} and in total the calculation of all superbubbles takes \mbox{$\Theta(n(n+m))$-time}.
Below, we show that, under a reasonable model, the algorithm takes constant time on average and thus all superbubbles can be found in $\Theta(n)$-time in total.

As we will see in the next section, although there are tens of thousands of superbubbles in practical unipath graphs, the entire graph is so large that its size is orders of magnitude greater than the total size of superbubbles.
Thus, most of the time spent in the iterated executions of the algorithm is dedicated for traversing regions that are far away from any superbubbles.
Therefore, it is reasonable to reduce the analysis of the algorithm to the evaluation of the time spent until the traversal of a non-superbubble region is aborted.
In such a case, if a vertex is not pushed into $S$ when it is labelled as \textsf{seen}, then it is very unlikely to be visited afterwards.
In other words, once the algorithm comes across a vertex of indegree greater than 1, then it almost never proceeds to traverse its descendants.
With these observations in mind, we model the way the tree of \textsf{visited} vertices grows in the algorithm by the following probabilistic tree generation process.
It starts from the root.
Each vertex is \textsf{good} with probability $p$.
A \textsf{good} vertex corresponds to a vertex of indegree 1.
If a vertex is \textsf{good}, it spawns $i$ children with probability $p_i$.
The theory of Galton-Watson branching processes~\cite{LyoPer12} tells that the expected number of vertices of depth $i$ is $\Theta(r^i)$ where $r:=p\sum_i{ip_i}$, i.e., the expected number of children of each vertex.
Therefore, if $r<1$ the expected size of the tree is $\Theta(\frac{1}{1-r})$, a constant.
For the unipath graph we constructed from human genome data, $r$ was about 0.77 where $p$ and $p_i$ were determined as the proportion of vertices with particular in/out-degree within all vertices.

\section{Experiment}
\label{sec:experiments}
\paragraph{Procedures}
We first constructed the succinct de Bruijn graph with parameter $k=27$ and $d=3$ for the read set SRX016231, which was derived by sequencing a human individual by an Illumina sequencer.
The length of each read is 100bp and the coverage is about 40.
Next, we constructed the unipath graph as described in subsection \ref{sec:graph}.
The resulting unipath graph consists of 107,154,751 vertices and 210,207,840 edges.
Last, we found all superbubbles in the unipath graph by the algorithm in section \ref{sec:algorithm}.

\paragraph{Results}
Table \ref{fig:histogram} is the histogram of the size of superbubbles where the size of a superbubble means the number of vertices in it.
The superbubbles of size 2 are omitted because they are ordinary bubbles.
The superbubble of Fig.~\ref{fig:superbubble} is of size 20 and this histogram tells, among other things, that there are hundreds of equally or more complex superbubbles.
On the other hand, what matters the most for the application to genome assembly problem is whether superbubbles really capture erroneous or repeat/mutation abundant regions, which topological complexity alone does not necessarily suggest.
One way to assess the relevance of a superbubble in this regard is to compare the length of paths in it where length of an edge is the length of the sequence represented by the edge.
Note that topologically close paths can have a variety of lengths because each edge can be originated from a unipath.
But among 23,078 superbubbles of size equal to or greater than 5 we found, 19,926 (86.3\%) of them have the longest/shortest path length ratio smaller than 1.05.
Therefore, superbubbles like that of Fig. \ref{fig:superbubble} are indeed typical.

\begin{table}[tb]
\begin{center}
\caption{Histogram of the size of superbubbles}
\label{fig:histogram}
\begin{tabular}{|c||c|c|c|c|c|c|c|}
\hline
size & 3-9 & 10-19 & 20-29 & 30-39 & 40-49 & 50-59 & 60- \\
\hline
\#S.B. & 71663 & 4295 & 347 & 69 & 21 & 8 & 3 \\
\hline
\end{tabular}
\end{center}
\end{table}

In terms of the computation time, it took 742.1 seconds for a Xeon 3.0GHz CPU to enumerate all superbubbles including ordinary bubbles.
The number of vertices visited was 126,537,254.

\section{Concluding Remarks}
\label{sec:conclusion}
We introduced the concept of superbubbles in assembly graphs, 
and proposed an efficient algorithm for detecting them. 
But many tasks remain as future work. 
It is an open problem whether it is possible to 
detect superbubbles in worst-case linear time.
Developing methods for categorizing the detected superbubbles
(e.g., errors, repeats, mutations, and polyploids), 
and methods for fixing errors in superbubbles are important future tasks. 
It is also interesting to extend our algorithm for other 
bubble-like structures (e.g. the bulge structure \cite{NurBan13}).

\subsection*{Acknowledgments}
KS and TS are supported in part by KAKENHI 23240002.
This research was supported by JST, ERATO, Kawarabayashi Large Graph Project.
The super-computing resource was provided in part 
by Human Genome Center, the Institute of Medical Science, the University of Tokyo.


\bibliographystyle{plain}
\bibliography{superbubble}



\end{document}